\documentclass[a4paper,USenglish, cleveref, autoref]{lipics-v2019}

\usepackage{bm}
\usepackage{here}
\usepackage{amsthm}
\theoremstyle{definition}
\usepackage{amsmath,amssymb}
\usepackage{algorithm}
\usepackage[noend]{algpseudocode}
\usepackage{caption}
\usepackage{url}
\usepackage{paralist}

\title{Storing Set Families More Compactly with Top ZDDs} 

\author{Kotaro Matsuda}{Graduate School of Information Technology and Science, The University of Tokyo, Japan}{kotaro\_matsuda@mist.i.u-tokyo.ac.jp}{}{}

\author{Shuhei Denzumi}{Graduate School of Information Technology and Science, The University of Tokyo, Japan}{denzumi@mist.i.u-tokyo.ac.jp}{https://orcid.org/xxxx-xxxx-xxxx-xxxx}{}

\author{Kunihiko Sadakane}{Graduate School of Information Technology and Science, The University of Tokyo, Japan}{sada@mist.i.u-tokyo.ac.jp}{https://orcid.org/0000-0002-8212-3682}{}

\authorrunning{K. Matsuda, S. Denzumi and K. Sadakane}

\Copyright{Kotaro Matsuda, Shuhei Denzumi and Kunihiko Sadakane}



\category{}

\relatedversion{}

\supplement{}

\acknowledgements{}

\nolinenumbers 

\hideLIPIcs  

\newcommand{\edgetypespan}{{\it type\_span}}
\newcommand{\lspan}{{\it label\_span}}
\newcommand{\edgetyperoot}{{\it type\_root}}
\newcommand{\edgetypein}{{\it type\_in}}
\newcommand{\srcroot}{{\it B_{src\_root}}}
\newcommand{\dstroot}{{\it dst\_root}}
\newcommand{\srcin}{{\it src\_in}}
\newcommand{\dstin}{{\it dst\_in}}
\newcommand{\dstdummy}{{\it dst\_dummy}}
\newcommand{\clustersize}{{\it clsize}}
\newcommand{\pdiff}{{\it preorder\_diff}}
\newcommand{\ldiff}{{\it label\_diff}}
\newcommand{\edgeBV}{{\it B_{edge}}}
\newcommand{\dummyBV}{{\it B_{dummy}}}
\newcommand{\horizontalBV}{{\it B_{H}}}

\begin{document}
\maketitle

\begin{abstract}
Zero-suppressed Binary Decision Diagrams (ZDDs) are data structures
for representing set families in a compressed form.
With ZDDs, many valuable operations on set families can be done 
in time polynomial in ZDD size.
In some cases, however, the size of ZDDs for representing large set families
becomes too huge to store them in the main memory.

This paper proposes top ZDD, a novel representation of ZDDs which uses
less space than existing ones.
The top ZDD is an extension of top tree, which compresses trees, to
compress directed acyclic graphs by sharing identical subgraphs.
We prove that navigational operations on ZDDs can be done in time poly-logarithmic
in ZDD size, and show that there exist set families for which
the size of the top ZDD is 
exponentially
smaller than that of the ZDD.
We also show experimentally that our top ZDDs have smaller size
than ZDDs for real data.
\end{abstract}



\section{Introduction}\label{intro}

Zero-suppressed Binary Decision Diagrams (ZDDs)~\cite{Minato}
are data structures which are derived from Binary Decision Diagrams (BDDs)~\cite{Bryant}
and which represent a family of sets (combinatorial sets) in a compressed form by 
Directed Acyclic Graphs (DAGs).
ZDDs are data structures specialized for processing set families and it is known
sparse set families can be compressed well.
ZDDs support binary operations between two set families in time polynomial to the ZDD size.
Because of these advantages, ZDDs are used for combinatorial optimization problems
and enumeration problems.

Though ZDDs can store set families compactly, their size may grow for some set families,
and we need further compression.
DenseZDDs~\cite{denseZDD} are data structures for storing ZDDs in a compressed form
and supporting operations on the compressed representation.
DenseZDDs represent a ZDD by a spanning tree of the DAG representing it,
and an array of pointers between nodes on the spanning tree.
Therefore its size is always linear to the original size
and to compress more, we need another representation.

Our basic idea for compression is as follows.
In a ZDD, the identical sub-structures are shared and replaced by pointers.
However identical sub-structures cannot be shared if they appear at different heights in ZDD.
As a result, even if the DAG of a ZDD contains repetitive structures in height direction,
they cannot be shared.

For not DAGs but trees, there exists a data structure called top DAG compression~\cite{Bille},
which can capture repetitive structures in height direction.
We extend it for DAGs and apply to compress ZDDs which support the operations on compressed ZDDs.

\subsection{Our contribution}

We propose top ZDDs, which partition the edges of a ZDD into a spanning tree
and other edges called complement edges, and store each of them in a compressed form.
For the spanning tree, we use the top DAG compression, which represents a tree by
a DAG with fewer number of nodes.  For the complement edges, we store them in some nodes
of the top DAG by sharing identical edges.
We show that basic operations on ZDDs can be supported in $O(\log^2n)$ time
where $n$ is the number of nodes of the ZDD.
For further compression we use succinct data structures for trees~\cite{Navarro14}
and for bitvectors~\cite{Raman07,Grossi05}.

We show experimental results on the size of our top ZDDs and existing data structures,
and query time on them.  The results show that the top ZDDs use less space for most of
input data.


%

\section{Preliminaries}
Here we explain notations and basic data structures.

Let $C = \{1, \ldots, c\}$ be the universal set.  Any set in this paper is a subset of $C$.
The empty set is denoted by $ \emptyset $.
For a set $S = \{a_1, \ldots, a_s\} \subseteq C \ (s \ge 1)$, its size is denoted by $|S| = s$.
The size of the empty set is $|\emptyset| = 0$.
A subset of the power set of $C$ is called a set family.
If a set family $\mathcal{F}$ satisfies either
$S \in \mathcal{F} \Rightarrow \forall k \in S,  S\backslash\{k\} \in \mathcal{F}$
or
$S\in \mathcal{F} \Rightarrow \forall k\in C, S\cup\{k\} \in \mathcal{F}$,
$\mathcal{F}$ is said to be monotone.
If the former is satisfied, $\mathcal{F}$ is monotone decreasing
and the latter monotone increasing.

\subsection{Zero-suppressed Binary Decision Diagrams}
Zero-suppressed Binary Decision Diagrams (ZDDs)~\cite{Minato}
are data structures for manipulating finite set families.
A ZDD is a directed acyclic graph (DAG) $G=(V, E)$ with a root node
satisfying the following properties.
A ZDD has two types of nodes; branching nodes and terminal nodes.
There are two types of terminal nodes $\bot$ and $\top$.
These terminal nodes have no outgoing edges.
Each branching node $v$ has an integer label $\ell(v) \in \{1, \ldots, c\}$,
and also has two outgoing edges $0$-edge and $1$-edge.
The node pointed to by the $0$-edge ($1$-edge) of $v$ is denoted
by $v_0 = \it{zero}(v)$ ($v_1 = \it{one}(v)$).
If for any branching node $v$ it holds $\ell(v) < \ell(v_0)$ and $\ell(v) < \ell(v_1)$,
the ZDD is said to be ordered.
In this paper, we consider only ordered ZDDs.
For convenience, we assume $\ell(v) = c + 1$ for terminal nodes $v$.
We divide the nodes of the ZDD into layers
$L_1, \ldots, L_{c+1}$ $(i = 1, \ldots, c + 1)$ according to the labels of the nodes.
Note that if $i \geq j$ there are no edges from layer $L_i$ to layer $L_j$.
The number of nodes in ZDD $G$ is denoted by $|G|$ and called the size of the ZDD.
On the other hand, the data size of a ZDD stands for the number of bits used in the data structure
representing the ZDD.

The set family represented by a ZDD is defined as follows.
\begin{definition}[The set family represented by a ZDD]\label{zdddef}
Let $v$ be a node of a ZDD and $v_0 = \it{zero}(v)$, $v_1 = \it{one}(v)$.
Then the set family $\mathcal{F}_v$
represented by $v$ is defined as follows.
\begin{enumerate}
\item If $v$ is a terminal node:
if $v = \top$, $\mathcal{F}_v = \{\emptyset\}$,
if $v = \bot$, $\mathcal{F}_v = \emptyset$.
\item If $v$ is a branching node:
$\mathcal{F}_v = \{S \cup \{\ell(v)\} \mid S \in \mathcal{F}_{v_1}\} \cup \mathcal{F}_{v_0}$.
\end{enumerate}
\end{definition}

For the root node $r$ of ZDD $G$,
$\mathcal{F}_r$ corresponds to the set family represented by the ZDD $G$.
This set family is also denoted by $\mathcal{F}_G$.

All the paths from the root to the terminal $\top$ on ZDD $G$
have one-to-one correspondence to all the sets $S = \{a_1, \ldots, a_{s}\}$ in the set family
represented by $G$.
Consider a traversal of nodes from the root towards terminals
so that for each branching node $v$ on the path, if $\ell(v) \not \in S$
we go to $v_0 = \it{zero}(v)$ from $v$, and if $\ell(v) \in S$
we go to $v_1 = \it{one}(v)$ from $v$.
By repeating this process, if $S\in \mathcal{F}_G$ we arrive at $\top$,
and if $S\not \in \mathcal{F}_G$ we arrive at $\bot$ or
the branching node corresponding to $a_i\in S$ does not exist.

\subsection{Succinct data structures}
Succinct data structures are data structures whose size match the information theoretic lower bound.
Formally, a data structure is succinct if any element of a finite set $U$ with cardinality $L$
is encoded in $\log_2(L)+o(\log_2(L))$ bits.
In this paper we use succinct data structures for bitvectors and trees.

\subsubsection{Bitvectors}

Bitvectors are the most basic succinct data structures.
A length-$n$ sequence $B \in \{0,1\}^n$ of $0$'s and $1$'s is called a bitvector.
On this bitvector we consider the following operations:
\begin{itemize}
\item $access(B,i)$ ($1 \le i \le n$): returns $B[i] \in \{0,1\}$, the $i$-th entry of $B$.
\item $rank_c(B,i)$ ($1 \le i \le n, c=0,1$): returns the number of $c$ in the first $i$ bits of $B$.
\item $select_c(B,j)$ ($1 \le j \le n, c=0,1$): returns the position of the $j$-th occurrence of $c$ in $B$.
\end{itemize}

The following result is known.
\begin{theorem}(\cite{Raman07})
For a bitvector of length $n$, using a $n+\mathrm{O}(n\log\log n/\log n)$-bit data structure
constructed in $\mathrm{O}(n)$ time,
$access(B,i), rank_c(B,i), select_c(B,j)$ are computed in constant time
on the word-RAM with word length $\Omega(\log n)$.
\end{theorem}

Consider a bitvector of length $n$ with $m$ ones.
For a sparse bitvector, namely, the one with $m={\mathrm{o}}(n/\log n)$,
we can obtain a more space-efficient data structure.
\begin{theorem}(\cite{Grossi05})
For a bitvector $B$ of length $n = 2^w$ with $m$ ones,
$select_1(B,i)$ is computed in constant time
on the word-RAM with word length $\Omega(\log n)$ using a
$m(2+w - \lfloor \log_2 n \rfloor)+{\mathrm{O}(m\log\log m/\log m)}$-bit data structure.
\end{theorem}
Note that on this data structure,
$rank_0,rank_1,select_0$ takes $\mathrm{O}(\log m)$ time.

\subsubsection{Trees}
Consider a rooted ordered tree with $n$ nodes.  An information-theoretic lower bound
of such trees is $2n-\Theta(\log n)$ bits.
We want to support the following operations:
(1) $parent(x)$: returns the parent of node $x$,
(2) $firstchild(x),lastchild(x)$: returns the first/last child of node $x$,
(3) $nextsibling(x),prevsibling(x)$: returns the next/previous sibling of node $x$
(4) $isleaf(x)$: returns if node $x$ is a leaf or not,
(5) $preorder\_rank(x)$: returns the preorder of node $x$,
(6) $preorder\_select(i)$: returns the node with preorder $i$,
(7) $leaf\_rank(x)$: returns the number of leaves whose preorders are smaller than that of node $x$,
(8) $leaf\_select(i)$: returns the $i$-th leaf in preorder,
(9) $depth(x)$: returns the depth of node $x$, that is, the distance from the root to $x$,
(10) $subtreesize(x)$: returns the number of nodes in the subtree rooted at node $x$,
(11) $lca(x,y)$: returns the lowest common ancestor (LCA) between nodes $x$ and $y$.

\begin{theorem}(\cite{Navarro14})
On the word-RAM with word length $\Omega(\log n)$,
the above operations are done in constant time using
a $2n+{\mathrm{o}}(n)$-bits data structure.
\end{theorem}
We call this the BP representation in this paper.


\subsection{DenseZDD}
A DenseZDD~\cite{denseZDD} is a static representation of a ZDD with attricuted edges~\cite{Minato90} by using some succinct data structures.
In comparison to the ordinary ZDD, a DenseZDD provides a much faster membership operation and less memory usage for most of cases.
When we construct a DenseZDD from a given ZDD, 
dummy nodes are inserted 
so that $\ell(v_0)=\ell(v)+1$ holds for each internal node $v$ for fast traversal.
The spanning tree consisting of all reversed 0-edges is represented by straight forward BP.
The DenseZDD is a combination of this BP and other succinct data structures that represent remaining information of the given ZDD.


\section{Top Tree and Top DAG}
We explain 
top DAG compression~\cite{Bille} to compress labeled rooted trees.

Top DAG compression is a compression scheme for labeled rooted trees
by converting the input tree into top tree~\cite{toptree} and then
compress it by DAG compression~\cite{Buneman03,Downey80,Frick03}.
DAG compression is a scheme to represent a labeled rooted tree by a smaller DAG obtained by merging identical subtrees of the tree.
Top DAG compression can compress repeated sub-structures (not only subtrees).
For example, a path of length $n$ with identical labels can be represented
by a top DAG with $O(\log n)$ nodes.
Also, for a tree with $n$ nodes, accessing a node label,
computing the subtree size,
and tree navigational operations such as first child and parent
are done in $O(\log n)$ time.
Here we explain the top tree and its greedy construction algorithm.
We also explain operations on top DAGs.

The top tree~\cite{toptree} for a labeled rooted tree $T$
is a binary tree $\mathcal{T}$ representing the merging process of clusters of $T$ defined as follows.
We assume that all edges in the tree are directed from the root towards leaves,
and an edge $(u,v)$ denotes the edge from node $u$ to node $v$.
Clusters are subsets of $T$ with the following properties.
\begin{itemize}
\item A cluster is a subset $F$ of the nodes of the original tree $T$ such that
nodes in $F$ are connected in $T$.
\item $F$ forms a tree and we regard the node in $F$ closest to the root of $T$ as the root of the tree.
We call the root of $F$ as the \emph{top boundary node},
\item $F$ contains at most one node having directed edges to outside of $F$.
If there is such a node, it is called the \emph{bottom boundary node}.
\end{itemize}
A boundary node is either a top boundary node or a bottom boundary node.

By merging two adjacent clusters, we obtain a new cluster,
where merge means to take the union of node sets of two clusters and make it as the node set of
the new cluster.
There are five types of merges, as shown in Figure~\ref{clustermerge}.
In the figure, ellipses are clusters before merge, black circles are boundary nodes of new clusters,
and white circles are not boundary nodes in new clusters.
\begin{figure}[hbtp]
\begin{center}
\includegraphics[clip,width=8.0cm]{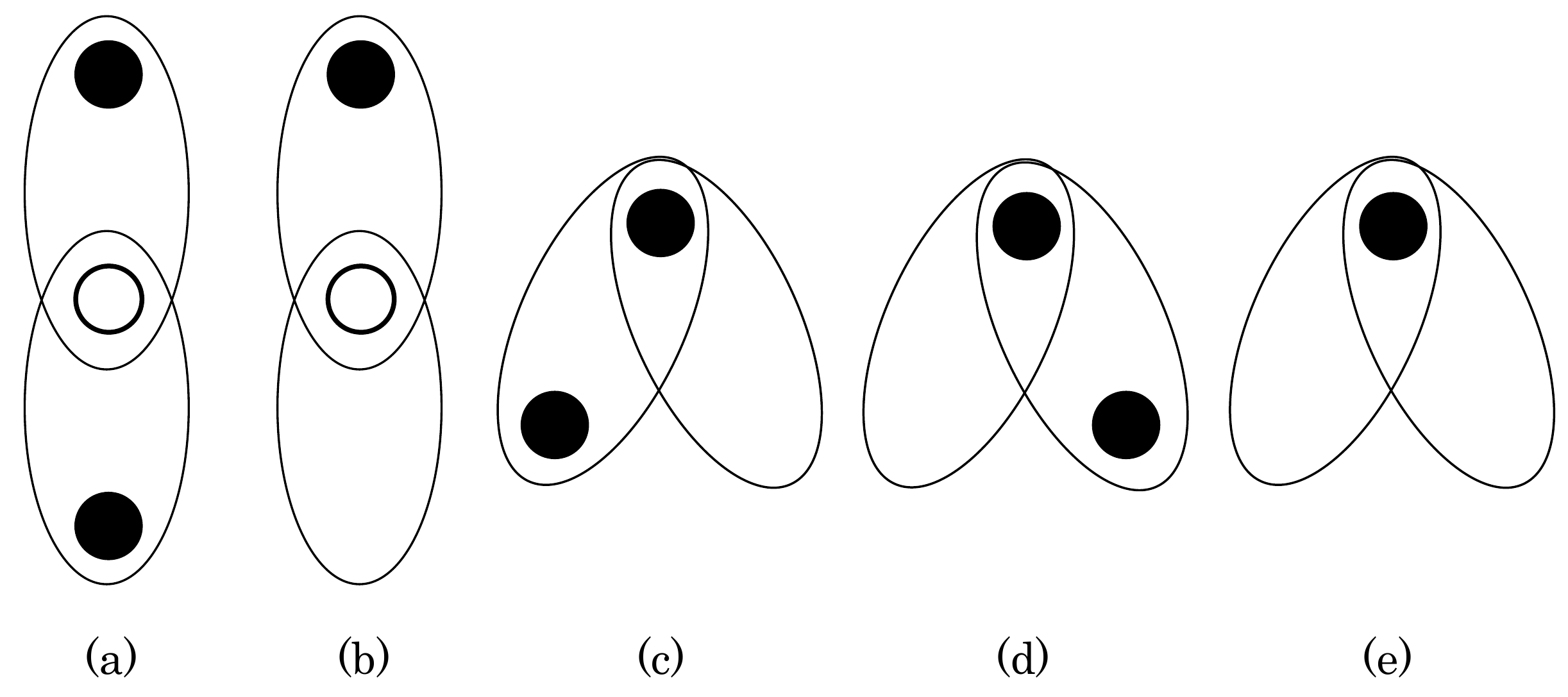}
\caption{Merging clusters.}
\label{clustermerge}
\end{center}
\end{figure}

These five merges are divided into two.
\begin{enumerate}
\item (a)(b) Vertical merge:
two clusters can be merged vertically if the top boundary node of one cluster
coincides with the bottom boundary node of the other cluster, and there are no edges
from the common boundary node to nodes outside the two clusters.
\item (c)(d)(e) Horizontal merge: 
two clusters can be merged horizontally if the top boundary nodes of the two clusters
are the same and at least one cluster does not have the bottom boundary node.
\end{enumerate}

The top tree of the tree $T$ is a binary tree $\mathcal{T}$ satisfying the following conditions.
\begin{itemize}
\item Each leaf of the top tree corresponds to a cluster with the endpoints of an edge of $T$.
\item Each internal vertex of the top tree corresponds to the cluster made by merging 
the clusters of its two children.  This merge is one of the five types
in Figure~\ref{clustermerge}.
\item The cluster of the root of the top tree is $T$ itself.
\end{itemize}


We call the DAG obtained by DAG compression of the top tree $\mathcal{T}$
as top DAG $\mathcal{T}D$, and the operation to compute the top DAG $\mathcal{T}D$
from tree $T$ is called top DAG compression~\cite{Bille}.

We define labels of vertices in the top tree to apply DAG compression as follows.
For a leaf of the top tree, we define its label as the pair of labels of both endpoints of
the corresponding edge in $T$.
For an internal vertex of the top tree, 
its label must have the information about cluster merge.
It is enough to consider three types of merges, not five as in Figure~\ref{clustermerge}.
For vertical merges, it is not necessary to store the information that the merged cluster
has the bottom boundary node or not.
For horizontal merges, it is enough to store if the left cluster has a bottom boundary node or not.
From this observation, we define labels of internal vertices as follows.
\begin{itemize}
\item For vertices corresponding to vertical merge: we set their labels as V.
\item For vertices corresponding to horizontal merge: we set their labels as
$\mathrm{H}_\mathrm{L}$ if the left child cluster has the bottom boundary node,
or $\mathrm{H}_\mathrm{R}$ if the right child cluster has the bottom boundary node.
If both children do not have bottom boundary nodes, the label can be arbitrary.
\end{itemize}

Top trees created by a greedy algorithm satisfy the following.
\begin{theorem}(\cite{Bille})
Let $n$ be the number of nodes of a tree $T$.
Then the height of $\mathrm{top\ tree}$ $\mathcal{T}$ created by a greedy algorithm
is $\mathrm{O}(\log n)$.
\end{theorem}


Consider to support operations on a tree $T$ which is represented
by top DAG $\mathcal{T}D$.
From now on, a node $x$ in $T$ stands for the node with preorder $x$ in $T$.
By storing additional information to each vertex of the top DAG,
many tree operations can be supported~\cite{Bille}.
For example, $Access(x)$ returns the label of $x$
and $Decompress(x)$ returns the subtree $T(x)$ of $T$ rooted at $x$.
For a tree with $n$ nodes, all operations except $Decompress(\cdot)$
are done in ${\rm O}(\log n)$ time, and $Decompress(\cdot)$ is done in
${\rm O}(\log n + |T(x)|)$ time.
Algorithm~\ref{access:topdag} shows a pseudo code.

\section{top ZDD}
We explain our top ZDD, which is a representation of ZDD by top DAG compression.
Though it is easy to apply our compression scheme for general rooted DAGs,
we consider only compression of ZDDs.

A ZDD $G=(V,E)$ is a directed acyclic graph in which nodes have labels $\ell(\cdot)$
(terminal nodes have $\bot$ and $\top$) and edges have labels $0$ or $1$.
We can regard it as a graph with only edges being labeled.
For each edge $(u,v)$ of ZDD $G$, we define its label
as a pair (edge label $0$/$1$, $\ell(u)-\ell(v)$) if $v$ is a branching node,
or a pair (edge label $0$/$1$, $\bot$/$\top$) if $v$ is a terminal node.
In practice, we can use $c+1$ instead of $\bot$ and $c+2$ instead of $\top$ for the second element,
where $c+1 = \ell(\bot) = \ell(\top)$.
Below we assume ZDDs have labels for only edges,
and $0$-edge comes before $1$-edge for each node.

Next we consider top trees for edge-labeled trees.
The difference from node-labeled trees is only how to store the information
for single edge clusters.
In top trees,
we stored labels for both endpoints of edges.
We change this for storing only edge labels.

The top ZDD is constructed from a ZDD $G=(V,E)$ as follows.
\begin{enumerate}
\item We perform a depth-first traversal from the root of $G$ and
obtain a spanning tree $T$ of all branching nodes.
During the process, we do not distinguish $0$-edges and $1$-edges, and terminal nodes
are not included in the tree.
Nodes of the tree are identified with their preorders in $T$.  If we say node $u$,
it means the node in $T$ with preorder $u$.
We call edges of $G$ not included in $T$ as \emph{complement edges}.
\item We convert the spanning tree $T$ to a top tree $\mathcal{T}$ by the greedy algorithm.
\item For each complement edge $(u,v)$, we store its information in a node of $\mathcal{T}$ as follows.
If $v$ is a terminal, let $a$ be the vertex of the top tree corresponding to the cluster
of single edge between $u$ and its parent in $T$.  Note that $a$ is uniquely determined.
Then we store a triple ($(u,v)$, edge label $0$/$1$, $\bot$/$\top$) in $a$.
If $v$ is a branching node, we store the information of the complement edge
to a vertex of $T$ corresponding to a cluster containing both $u$ and $v$.
The information to store is a triple ($(u,v)$, edge label $0$/$1$, $\ell(u)-\ell(v)$).
We decide a vertex to store it as follows.  Let $a, b$ be the vertices of the top tree
corresponding to the clusters of single edges towards $u, v$ in $T$, respectively.
Then we store the triple in the lowest common ancestor ${\mathrm{lca}}(a,b)$ in $T$.
Here the information $(u,v)$ represents local preorders inside the cluster
corresponding to ${\mathrm{lca}}(a,b)$.
Note that ${\mathrm{lca}}(a,b)$ may not be the minimal cluster including both $u$ and $v$.
\item We create a top DAG $\mathcal{T}D$ by DAG compression by sharing identical clusters.
To determine identicalness of two clusters, we compare them together with the information
of complement edges in the clusters stored in step 3.
Complement edges which do not appear in multiple clusters are moved to the root of $T$.
\end{enumerate}
Figure~\ref{example:topzdd} shows an example of a top ZDD.
The left is the original ZDD and the right is the corresponding top ZDD.
Red and green edges show edges in the spanning tree and complement edges, respectively.
In this figure we show for each vertex of the top DAG, the corresponding cluster,
but they are not stored explicitly.

\begin{figure}[hbtp]
\begin{center}
\includegraphics[clip,width=8.0cm]{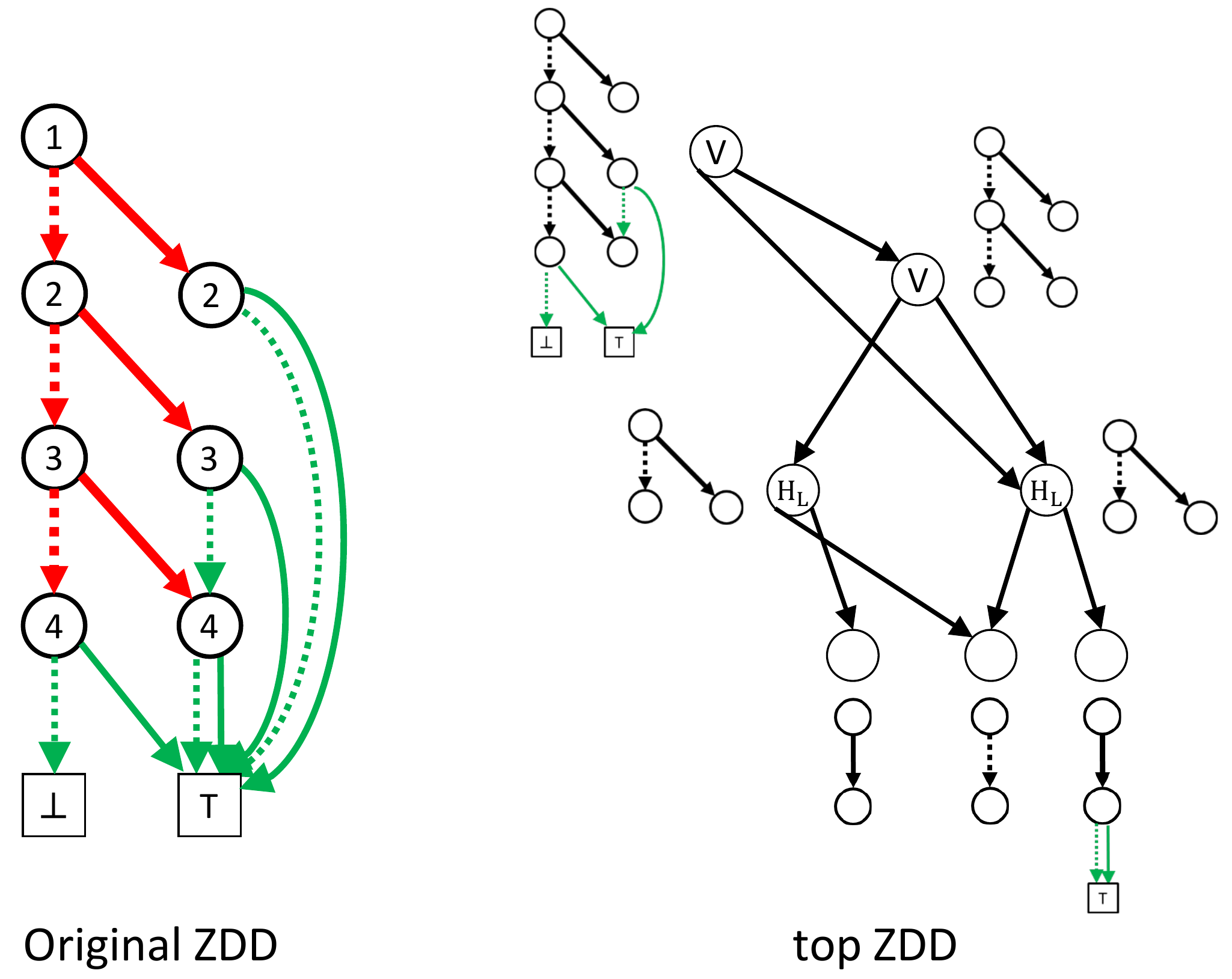}
\caption{An example of a top ZDD.  
Terminal nodes and branching nodes are depicted by squares and circles, respectively, and $0$-edges and $1$-edges are depicted by dotted and solid lines, respectivelty.
Red edges are spanning tree edges and green edges are complement edges.
For each vertex of the top DAG, the corresponding cluster and the information stored in the vertex
are shown.}
\label{example:topzdd}
\end{center}
\end{figure}

To achieve small space, it is important to use what data structure for representing each information.
For example, 
we explained that each vertex of the top DAG
stores the cluster size etc., this is redundant and the space can be reduced.
Next we explain our space-efficient data structure which is enough to support efficient queries in detail.

\subsection{Details of the data structure}
We need the following information to recover the original ZDD from a top ZDD.
\begin{itemize}
\item Information about the structure of top DAG $\mathcal{T}D$.
\item Information about each vertex of $\mathcal{T}D$.
There are three types of vertices: vertices corresponding to a leaf of the top tree,
vertices representing vertical merge, and vertices representing horizontal merge.
For each type we store different information.
\item Information about complement edges.  They are stored in the root or other vertices of $\mathcal{T}D$.
\end{itemize}
We show space-efficient data structures for storing these information.
In theory, we use the succinct bitvector with constant time rank/select support~\cite{Raman07}.
In practice, we use the SparseArray~\cite{OkanoharaS07} to compress a bitvector
if the ratio of ones in the bitvector is less than $\frac{1}{4}$,
and use the SparseArray for the bitvector whose 0/1 are flipped if
the ratio of zeros is less than $\frac{1}{4}$.
To store an array of non-negative integers, we use $\lfloor\log_2 m\rfloor$ bits for each entry
where $m$ is the maximum value in the array.
Let $n$ denote the number of internal nodes of a ZDD.
We use $n+1,n+2$ to represent terminals $\bot,\top$, respectively.

\subsubsection{The data structure for the structure of top DAG $\mathcal{T}D$}
We store top DAG $\mathcal{T}D$ after converting it to a tree.
We make tree $T'$ by adding dummy vertices to $\mathcal{T}D$.
For each vertex $x$ of $\mathcal{T}D$ whose in-degree is two or more, we do the following.
\begin{enumerate}
\item Let $a_1,\cdots,a_t$ be the vertices of $\mathcal{T}D$ from which
there are edges towards $x$.  Note that there may exist identical vertices among them
corresponding to different edges.
We create $t-1$ dummy vertices $d_1,\cdots,d_{t-1}$.
\item For each $1\leq i \leq t-1$, remove edge $(a_i,x_i)$ and add edge $(a_i,d_i)$.
\item For each dummy vertex $d_i$, we store information about a pointer to $x$.
In our implementation, we store the preorder of $x$ in $T'$ from which the dummy vertices are removed.
\end{enumerate}
Then we can represent the structure of the top DAG by the tree $T'$ and the pointers from the dummy vertices.

Next we explain how to store $T'$ and the information about the dummy vertices.
The structure of $T'$ is represented by the BP sequence~\cite{Navarro14}.
There are two types of leaves in $T'$: those which exist in the original top DAG,
and those for the dummy vertices.  To distinguish them, we use a bitvector.
Let $m$ be the number of leaves in $T'$.  We create a bitvector $\dummyBV$ of length $m$
whose $i$-th bit corresponds to the $i$-th leaf of $T'$ in preorder.
We set $\dummyBV[i] =1$ if the $i$-th leaf is a dummy vertex, 
and we set $\dummyBV[i] = 0$ otherwise.

We add additional information to dummy vertices to support efficient queries.
We define an array $\clustersize$ of length $D$ where $D$ is the number of dummy vertices.
For the $i$-th dummy vertex in preorder, let $s_i$ be the vertex pointed to by the dummy vertex.
We define $\clustersize[k] = \sum_{i=1}^{k}(\mbox{the number of vertices in the cluster represented by} s_i)$.
That is, $\clustersize[k]$ stores the cumulative sum of cluster sizes up to $k$.
This array is used to compute the cluster size for each vertex efficiently.

\subsubsection{Information on vertices}
We explain how to store information on vertices of $T'$ except for dummy vertices.

Each vertex corresponding to a leaf in the original top tree
is a cluster for a single edge in the spanning tree, and it is a non-dummy leaf in $T'$.
We sort these vertices in preorder in $T'$,
and store information on edges towards them in the following two arrays.
One is array $\lspan$ to store differences of levels between endpoints of edges.
Let $u$ and $v$ be the starting and the ending points of the edge corresponding to the $i$-th leaf,
respectively.  Then we set $\lspan[i] = \ell(v)-\ell(u)$.
The other is array $\edgetypespan$ to store if an edge is $0$-edge or $1$-edge.
We set $\edgetypespan[i] = 0$ if the edge corresponding to the $i$-th vertex is a $0$-edge,
and $\edgetypespan[i] = 1$ otherwise.

Each vertex of $T'$ corresponding to vertical merge or horizontal merge
is an internal vertex.
We sort internal vertices of $T'$ in their preorder.
Then we make a bitvector $\horizontalBV$ so that
$\horizontalBV[i] = 0$ if the $i$-th vertex stands for vertical merge,
and $\horizontalBV[i] = 1$ if it stands for horizontal merge.
For vertices corresponding to horizontal merge, we do not store additional information.
For vertices corresponding to vertical merge,
we use arrays $\pdiff$ and $\ldiff$ to store the differences of preorders and levels between
the top and the bottom boundary nodes of the merged cluster.
Let $x_i$ be the $i$-th vertex in preorder corresponding to vertical merge,
$cl_i$ be the cluster corresponding to $x_i$,
$t_i$ be the top boundary node of $cl_i$, and $b_i$ be the bottom boundary node of $cl_i$.
Note that $t_i$ and $b_i$ are nodes of the ZDD.
Then we set $\pdiff[i] = (\mbox{the local preorder of $b_i$ inside cluster $cl_i$})$
and $\ldiff[i] = \ell(b_i) -\ell(t_i)$.

\subsubsection{Information on complement edges}
Complement edges are divided into two: those stored in the root of the top DAG
and those stored in other vertices.  We represent them in a different way.

First we explain the data structure for storing complement edges in the root of the top DAG.
Let $E_{root}$ be the set of all complement edges stored in the root.
We sort edges of $E_{root}$ in the preorder of their starting point.
Orders between edges with the same starting point are arbitrary.

For complement edges stored in the root,
we store the preorders of their starting point using a bitvector $\srcroot$,
the preorders of their ending point using an array $\dstroot$,
and edge labels $0$/$1$ using an array $\edgetyperoot$.
The cluster corresponding to the root of top DAG is the spanning tree of the ZDD.
For each node $v$ of the spanning tree, we represent the number of complement edges in $E_{root}$
whose starting point is $v$, using a unary code.  We concatenate and store them
in preorder in the bitvector $\srcroot$.
For edges in $E_{root}$ sorted in preorder of the starting points,
we store the preorder of the ending point of the $i$-th edge in $\dstroot[i]$,
and set $\edgetyperoot[i] = 0$ if the $i$-th edge is a $0$-edge,
and set $\edgetyperoot[i] = 1$ otherwise.

Next we explain the data structure for storing complement edges in vertices other than the root.
Let $E_{in}$ be the set of those edges.
We sort the edges as follows.
\begin{enumerate}
\item We divide the edges of $E_{in}$ into groups based on the clusters having the edges.
These groups are sorted in preorder of vertices for the clusters.
\item Inside each cluster $cl(x)$, we sort the edges of $E_{in}$ in preorder
of starting points of the edges.
For edges with the same starting point, their order is arbitrary.
\end{enumerate}
We store the sorted edges of $E_{in}$ using a bitvector $\edgeBV$ and
three arrays $\srcin$, $\dstin$, and $\edgetypein$.
The bitvector $\edgeBV$ stores the numbers of complement edges in vertices of $T'$
by unary codes.
The arrays $\srcin$, $\dstin$, and $\edgetypein$ are defined as:
$\srcin[i] = (\mbox{the local preorder of the starting point of the $i$-th edge
inside the cluster})$,
$\dstin[i] = (${the local preorder of the ending point of the $i$-th edge
inside the cluster}$)$,
$\edgetypein[i] = 0$ if the $i$-th edge is a $0$-edge,
and
$\edgetypein[i] = 1$ otherwise.

Table~\ref{datastructure} summarizes the components of the top ZDD.

\subsection{The size of top ZDDs}
The size of top ZDDs heavily depends on not only the number of vertices in the spanning tree
after top DAG compression, but also the number of complement edges for which we store some information.
Therefore the size of top ZDDs becomes small if the number of nodes is reduced by
top DAG compression and many common complement edges are shared.

In the best case, top ZDDs are exponentially smaller than ZDDs.
\begin{theorem}
There exists a ZDD with $n$ nodes to which the corresponding top ZDD has
$O(\log n)$ vertices.
\end{theorem}
\begin{proof}
A ZDD storing a power set with $n=2^m$ elements satisfies the claim.
Figure~\ref{besttopzdd} shows this ZDD and top ZDD.
A spanning tree of the ZDD is a path of $2^m$ many $0$-edges.
Its top tree has a leaf corresponding to a $0$-edge of length $1$, 
and internal vertices form a complete binary tree with height $m$.
If we apply DAG compression to this top tree, we obtain the DAG of length $m$ as
shown in Figure~\ref{besttopzdd}.
Sharing complement edges also works very well.
The $k$-th vertex from below representing a vertical merge stores
a $1$-edge connecting a node with local preorder $2^{k-1}+1$ inside a cluster
and a node with local preorder $2^{k-1}+2$.
The root of the top DAG stores $0$-edge and $1$-edge to the terminal $\top$.
Because the height of the top DAG is $O(\log n)$, the claim holds.
\end{proof}
\begin{figure*}[hbtp]
\begin{center}
\includegraphics[clip,width=4.0cm]{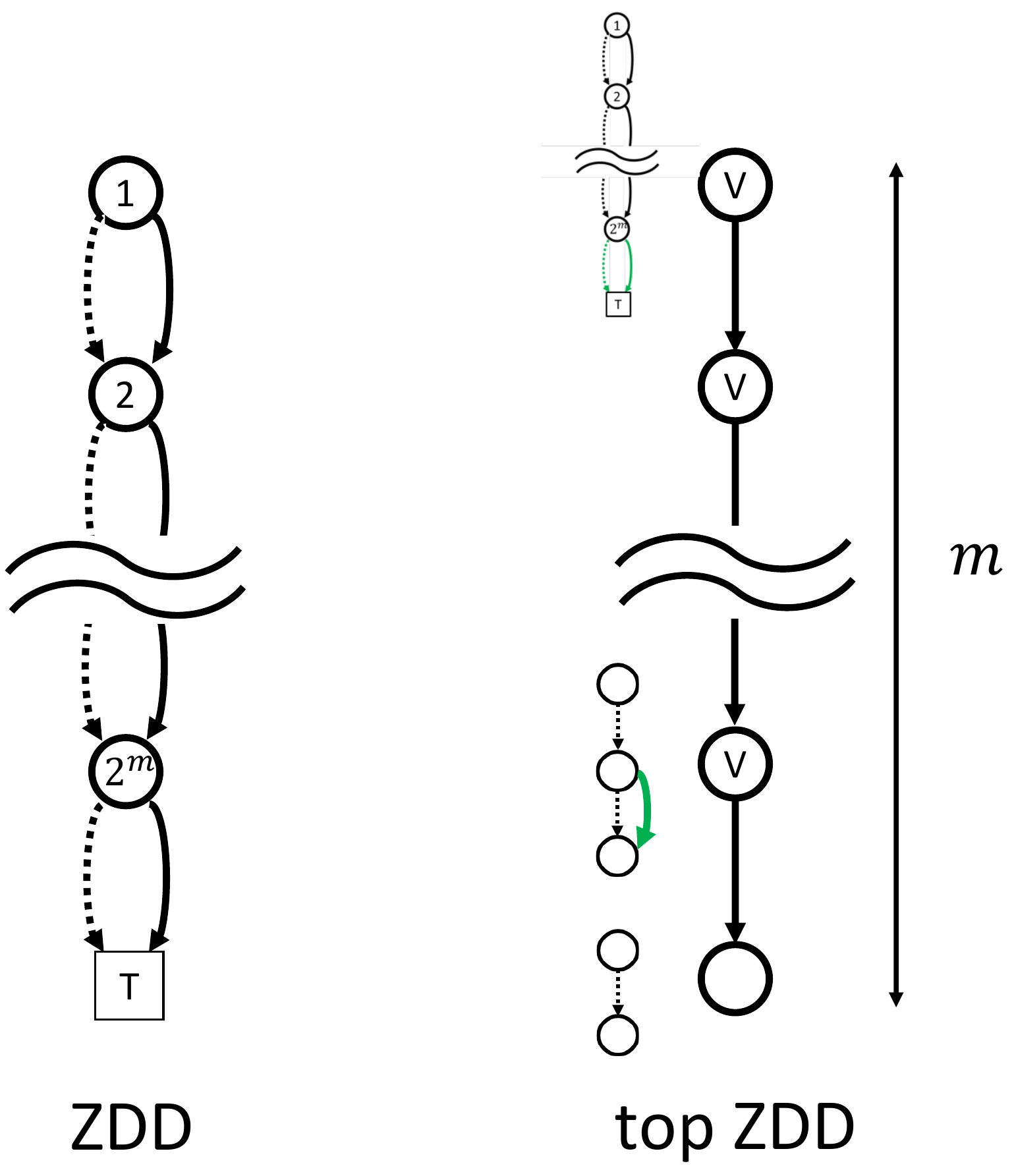}
\caption{A top ZDD with $O(\log n)$ vertices, where $n=2^m$.}
\label{besttopzdd}
\end{center}
\end{figure*}

\subsection{Operations on top ZDDs}
We give algorithms for supporting operations on the original ZDD using the top ZDD.
We consider the following three basic operations.
We identify a node $x$ of the ZDD with the vertex in the spanning tree $T$
used to create the top ZDD whose preorder is $x$.
\begin{itemize}
\item $\ell(x)$: returns the label of a branching node $x$.
\item $zero(x)$: returns the preorder of the node pointed to by the $0$-edge of $x$,
or returns $\bot$ or $\top$ if the node is a terminal.
\item $one(x)$:  returns the preorder of the node pointed to by the $1$-edge of $x$,
or returns $\bot$ or $\top$ if the node is a terminal.
\end{itemize}
We show $\ell(x)$ is done in $O(\log n)$ time and
other operations are done in $O(\log^2 n)$ time where $n$ is the
number of nodes of the ZDD.
Below we denote the vertex of $T'$ stored in the top ZDD with preorder $x$
by ``vertex $x$ of $T'$''.

First we explain how to compute $\ell(x)$ in $O(\log n)$ time.
We can compute $\ell(x)$ recursively using a similar algorithm to those
on the top DAG.  A difference is that
we assumed that each vertex of the top DAG stores the cluster size,
while in the top ZDD it is not stored to reduce the space requirement.
Therefore we have to compute it using the information in Table~\ref{datastructure}.

To work the recursive computation, we need to compute the cluster size $size(x')$
represented by vertex $x'$ of $T'$ efficiently.
We can compute $size(x')$ by the number of non-dummy leaves in the subtree of $T'$
rooted at $x'$, and the sizes of the clusters corresponding to dummy leaves
in the subtree rooted at $x'$.
If we merge two clusters of size $a$ and $b$, the resulting cluster has size $a+b-1$.
Therefore if we merge $k$ clusters whose total size is $S$, the resulting cluster
after $k-1$ merges has size $S-k+1$.
These values can be computed from the BP sequence $bp$ of $T'$,
the array $\clustersize$, and the bitvector $\dummyBV$.
By using $bp$, we can compute the interval $[l, r]$ of leaf ranks in the subtree
rooted at $x'$.
Then, using $\dummyBV$, we can find the number $c$ of non-dummy leaves
and the interval $[l',r']$ of non-dummy leaf ranks, in the subtree of $x'$.
Because $\clustersize$ is the array for storing cumulative sums of cluster sizes
for dummy leaves, the summation of sizes of clusters corresponding to
$l'$-th to $r'$-th dummy leaves is obtained from $\clustersize[r']-\clustersize[l'-1]$.
Because the size of a cluster for a non-dummy leaf is always $2$,
the summation of cluster sizes for non-dymmy leaves is also obtained.
Algorithm~\ref{size(x')} gives a pseudo code for computing $size(x')$.
This can be done in constant time.

Using the function $size(x')$, we can compute a recursive function similar to Algorithm~\ref{access:topdag}.
Instead of $D(\cdot)$ in Algorithm~\ref{access:topdag}, we use $\pdiff$.
When we arrive at a dummy leaf, we use a value in $\dstdummy$
to move to the corresponding internal vertex of $T'$ and restart the recursive computation.
Then for the vertex of the original ZDD whose preorder in $T$ is $x$,
we can obtain the leaf of $T'$ corresponding to the cluster of a single edge
containing $x$.

To compute $\ell(x)$, we traverse the path from the root of $T'$ to
the leaf corresponding to the cluster containing $x$.
First we set $s=1$.
During the traversal, if the current vertex is for vertical merge and
the next vertex is its right child, that is, the next cluster is in the bottom,
we add the $\ldiff$ value of the top cluster to $s$.
The index of $\ldiff$ is computed from $\horizontalBV$ and $bp$.
When we reach the leaf $p'$ of $T'$, if $x$ is its top boundary node,
it holds $\ell(x) = s$, otherwise, let $k = leaf\_rank(p')$, then
we obtain $\ell(x) = s + \lspan[k-rank_1(\dummyBV,k)]$.
Because each operation is done in constant time and the height of the top DAG is
$O(\log n)$, $\ell(x)$ is computed in $O(\log n)$ time.

Next we show how to compute $y = zero(x)$. We can compute $one(x)$ in a similar way.
We do a recursive computation as operations on top DAG,
A difference is how to process complement edges.
There are two cases: if the $0$-edge from $x$ is in the spanning tree or not.
If the $0$-edge from $x$ is in the spanning tree, the edge is stored
in a cluster with a single edge $(x, y)$.  The top boundary node of such a cluster is $x$.
Therefore we search clusters whose top boundary node is $x$.
If the $0$-edge from $x$ is not in the spanning tree, it is a complement edge and
it is stored in some vertex on the path from a cluster $C$ with a single edge
whose bottom boundary node is $x$ to the root.  Therefore we search for $C$.

%
First we recursively find a non-dummy leaf of $T'$ whose top boundary node is $x$.
During this process, if there is a vertex whose top boundary is $x$ and its cluster contains more than
one edge and corresponds to horizontal merge, we move to the left child, because
the $0$-edge from $x$ must exist in the left cluster.
If we find a non-dummy leaf of $T'$ which corresponds to a cluster with a single edge
and its top boundary node is $x$, its bottom boundary node is $y = zero(x)$.
We climb up the tree until the root to compute the global preorder of $y$.
If there does not exist such a leaf, the $0$-edge from $x$ is not in the spanning tree.
We find a cluster with a single edge whose bottom boundary node is $x$.
%
From the definition of the top ZDD, the $0$-edge from $x$ is stored in some vertices
visited during the traversal.
Because complement edges stored in a cluster are sorted in local preorders inside the cluster
of starting points, we can check if there exists a $0$-edge whose starting point is $x$
in $O(\log n)$ time.
If it exists, we obtain the local preorder of $y$ inside the cluster.
By going back to the root, we obtain the global preorder of $y$.
Note that complement edges for all clusters are stored in one array, and therefore we need
to obtain the interval of indices of the array corresponding to a cluster.
This can be done using $\edgeBV$.
In the worst case, we perform a binary search in each cluster on the search path.
Therefore the time complexity of $zero(x)$ is $O(\log^2 n)$.

\section{Experimental Comparison}
We compare our top ZDD with existing data structures.
We implemented top ZDD with C++ and measured the required space for storing the data structure.
For comparison, we used the following three data structures.
\begin{itemize}
\item top ZDD (proposed): we measured the space for storing the data structures in Table~\ref{datastructure}.
\item DenseZDD~\cite{denseZDD}: data structures for representing ZDDs using succinct data structures.
Two data structures are proposed; one support constant time queries and the other has
$O(\log n)$ time complexity.  We used the latter that uses less space.
\item a standard ZDD: a data structure which naively represents ZDDs.  We store for each node
its label and two pointers corresponding to a $0$-edge and a $1$-edge.
The space is $2n\lfloor\log n\rfloor + n\lfloor\log c\rfloor$ bits where $n$ is the number of nodes
of a ZDD and $c$ is the size of the universe of a set family.
\end{itemize}

We constructed ZDDs of the following set families.
\begin{itemize}
\item The power set of a set $\{1,..,A\}$ with $A$ elements.
\item For the set $\{1,..,A\}$ with $A$ elements,
the family of all the set $S$ satisfying \\
$(\mbox{The maximum value of $S$})-(\mbox{The minimum value of $S$})\leq B$.
\item For the set $\{1,..,A\}$ with $A$ elements,
the family of all the sets with cardinality at most $B$.
\item Knapsack set families with random weights.  That is,
for $i$-th element in a set ($1\leq i \leq A$), we define its weight $w_i$
as a uniformly random integer in $[1,W]$, then sort the elements in decreasing order of weights,
and construct a set family consisting of all sets with weight at most $C$.
\item The family of edge sets which are matching of a given graph.
As for graphs, we used the $8\times 8$ grid graph,
the complete graph with $12$ vertices $K_{12}$,
and a real communication network {\itshape ``Interoute''}.
\item Set families of frequent item sets.
\item Families of edge sets which are paths from the bottom left vertex to the top right vertex
in $n\times n$ grid graph, for $n=6,7,8,9$.
\item Families of solutions of the $n$-queen problem, for $n=11,12,13$.
\end{itemize}
We used several values for the parameters $A,B,C,W$.
The results are shown in Tables~\ref{powerset} to~\ref{top:nqueen}.
The unit of size is bytes.

\begin{table}[hbtp]
\small
\caption {The power set of $\{1,..,A\}$.}
\label{powerset}
\centering
\begin{tabular}{|c|c|c|c|}
\hline
        & top ZDD        & DenseZDD & $(2n\lfloor\log n\rfloor + n \lfloor \log c \rfloor)/8$ \\ \hline
$A=1000$ & \textbf{2,297} & 4,185    & 3,750                                                 \\ \hline
$A=50000$ & \textbf{2,507} & 178,764  & 300,000                                               \\ \hline
\end{tabular}
\end{table}

\begin{table}[hbtp]
\small
\caption {For the set $\{1,..,A\}$ with $A$ elements,
the family of all the set $S$ satisfying 
$(\mbox{The maximum value of $S$})-(\mbox{The minimum value of $S$})\leq B$.}
\label{limitedwidth}
\centering
\begin{tabular}{|c|c|c|c|}
\hline
             & top ZDD & DenseZDD & $(2n\lfloor\log n\rfloor + n \lfloor \log c \rfloor)/8$ \\ \hline
$A=500,B=250$  & \textbf{2,471}   & 227,798  & 321,594                                               \\ \hline
$A=1000,B=500$  & \textbf{2,551}   & 321,594  & 1,440,375                                             \\ \hline
\end{tabular}
\end{table}

\begin{table}[hbtp]
\small
\caption{For the set $\{1,..,A\}$ with $A$ elements,
the family of all the sets with cardinality at most $B$.}
\label{limitedsize}
\centering
\begin{tabular}{|c|c|c|c|}
\hline
               & top ZDD & DenseZDD & $(2n\lfloor\log n\rfloor + n \lfloor \log c \rfloor)/8$ \\ \hline
$A=100,B=50$   & \textbf{3,863}   & 9,544    & 9,882                                                 \\ \hline
$A=400,B=200$  & \textbf{13,654}  & 146,550  & 206,025                                               \\ \hline
$A=1000,B=500$ & \textbf{43,191}  & 966,519  & 1,440,375                                             \\ \hline
\end{tabular}
\end{table}

\begin{table}[hbtp]
\small
\caption{Knapsack set families with random weights.
$A$ is the number of elements, $W$ is the maximum weight of an element, $C$ is the capacity of the knapsack.}
\label{knapsack}
\centering
\begin{tabular}{|c|c|c|c|}
\hline
                     & top ZDD            & DenseZDD  & $(2n\lfloor\log n\rfloor + n \lfloor \log c \rfloor)/8$ \\ \hline
$A=100,W=1000,C=10000$ & \textbf{1,659,722} & 1,730,401 & 2,444,405                                               \\ \hline
$A=200,W=100,C=5000  $ & \textbf{1,032,636} & 1,516,840 & 2,181,688                                               \\ \hline
$A=1000,W=100,C=1000 $ & \textbf{2,080,965} & 2,929,191 & 4,491,025                                               \\ \hline
$A=5000,W=100,C=200 $  & \textbf{1,135,653} & 1,740,841 & 2,884,279                                               \\ \hline
$A=1000,W=10,C=1000$   & \textbf{1,383,119} & 2,618,970 & 3,990,350                                               \\ \hline
$A=1000,W=100,C=1000$  & \textbf{565,740}   & 656,728   & 1,056,907                                               \\ \hline
\end{tabular}
\end{table}

\begin{table}[hbtp]
\small
\caption{The family of edge sets which are matching of a given graph.}
\label{top:matching}
\centering
\begin{tabular}{|c|c|c|c|}
\hline
 & top ZDD & DenseZDD & $(2n\lfloor\log n\rfloor + n \lfloor \log c \rfloor)/8$ \\ \hline
$8\times 8$ grid & \textbf{12,246} & 16,150 & 18,014 \\ \hline
complete graph $K_{12}$ & 23,078 & \textbf{16,304} & 25,340 \\ \hline
“{\itshape{Interoute}}” & \textbf{30,844} & 39,831 & 50,144 \\ \hline
\end{tabular}
\end{table}

\begin{table}[hbtp]
\small
\caption{Set families of frequent item sets.}
\label{top:frequent}
\centering
\begin{tabular}{|c|c|c|c|}
\hline
 & top ZDD & DenseZDD & $(2n\lfloor\log n\rfloor + n \lfloor \log c \rfloor)/8$ \\ \hline
“{\itshape{mushroom}}” $(p=0.001)$ & 104,774 & \textbf{91,757} & 123,576 \\ \hline
“{\itshape{retail}}” $(p=0.00025)$ & \textbf{59,894} &{65,219} & 62,766 \\ \hline
“{\itshape{T40I10D100K}}” $(p=0.005)$ & \textbf{177,517} & 188,400 & 248,656 \\ \hline
\end{tabular}
\end{table}

\begin{table}[hbtp]
\small
\caption{Families of paths in $n\times n$ grid graph.}
\label{top:gridpath}
\centering
\begin{tabular}{|c|c|c|c|}
\hline
 & top ZDD & DenseZDD & $(2n\lfloor\log n\rfloor + n \lfloor \log c \rfloor)/8$ \\ \hline
$n=6$ & \textbf{17,194} & 28,593 & 37,441 \\ \hline
$n=7$ & \textbf{49,770} & 107,529 & 143,037 \\ \hline
$n=8$ & \textbf{157,103} & 401,251 & 569,908 \\ \hline
$n=9$ & \textbf{503,265} & 1,465,984 & 2,141,955 \\ \hline
\end{tabular}
\end{table}

\begin{table}[hbtp]
\small
\caption{Families of solutions of the $n$-queen problem.}
\label{top:nqueen}
\centering
\begin{tabular}{|c|c|c|c|}
\hline
 & top ZDD & DenseZDD & $(2n\lfloor\log n\rfloor + n \lfloor \log c \rfloor)/8$ \\ \hline
$n=11$  & 40,792 & \textbf{35,101} & 45,950 \\ \hline
$n=12$ & 183,443 &\textbf{167,259} & 229,165 \\ \hline
$n=13$ & 866,749 & \textbf{799,524} & 1,126,295 \\ \hline
\end{tabular}
\end{table}

We found that for all data sets, the top ZDD uses less space
than the naive representation of the standard ZDD.
We also confirmed that
the data sets in Tables~\ref{powerset}, \ref{limitedwidth}, and~\ref{limitedsize}
can be compressed very well by top ZDDs.
Table~\ref{knapsack} shows the results on the sets of solutions of knapsack problems.
For any case, the top ZDD uses less space than the DenseZDD,
and for some cases the memory usage of the top ZDD is almost the half of that of the DenseZDD.
Tables~\ref{top:matching} and~\ref{top:frequent} show the results for families of matching in a graph
and frequent item sets, respectively.  There are a few case that the DenseZDD uses less space
than the top ZDD.

The results above are for monotone set families, that is, any subset of a set a the family
also exists in the family.  Tables~\ref{top:gridpath} and~\ref{top:nqueen} show results
on non-monotone set families.
For the set of edges on the path from the bottom left corner to the top right corner
of an $n\times n$ grid graph, the top ZDD uses less space than the DenseZDD,
and for $n=9$, the top ZDD uses about $1/3$ the memory of DenseZDD.
On the other hand, for the sets of all the solutions of the $n$-queen problem,
the top ZDD uses about 10 \% more space than the DenseZDD.
From these experiments we confirmed that the top ZDD uses less space than the DenseZDD
for many set families.

Next we show construction time and edge traverse time of the top ZDD and the DenseZDD
in Tables~\ref{time:powerset} to~\ref{time:nqueen}.
For edge traverse time, we traversed from the root of a ZDD towards terminals
by randomly choosing $0$- or $1$-edge 65,536 times, and took the average.
When we arrived at a terminal, we restarted from the root.

The results show the DenseZDD is faster than the top ZDD for construction and traverse,
except for the construction time for the data set {\itshape{retail}}.
The traverse algorithm of the top ZDD is recursive and in the worst case
it recurses $\Theta(\log n)$ times, whereas that for the DenseZDD is not recursive.

\section{Concluding Remarks}
We have proposed top ZDD to compress a ZDD by regarding it as a DAG.
We compress a spanning tree of a ZDD by the top DAG compression,
and compress other edges by sharing them as much as possible.
We showed that the size of a top ZDD can be logarithmic of that of the standard ZDD.
We also showed that navigational operations on a top ZDD are done in time polylogarithmic
to the size of the original ZDD.
Experimental results show that the top ZDD always uses less space than the standard ZDD,
and uses less space than the DenseZDD for most of the data.

Future work will be as follows.
First, in the current construction algorithm, we create a spanning tree of ZDD by a depth-first search,
but this may not produce the smallest top ZDD.  For example, if we choose all $0$-edges,
we obtain a spanning tree whose root is the terminal $\top$, and this might be better.
Next, in this paper we considered only traversal operations and did not give advanced operations
such as choosing the best solution among all feasible solutions based on an objective function.
Lastly, we considered only compressing ZDDs, but our compression algorithm can be
used for compressing any DAG.  We will find applications of our compression scheme.

\bibliographystyle{abbrv}
\bibliography{hoge.bib}

\appendix

\newpage
\section{Pseudo Codes}
\begin{algorithm}[h]
\caption{$Access(x)$: computes the label of a node whose preorder in the tree representing the top DAG is $x$.}
\label{access:topdag}
\begin{algorithmic}[1]
\Require{Preorder $x$}
\Ensure{The label of node $x$}
\State{$r \gets $ the root of top DAG}
\State{\Return{\textsc{sub}$(r,x)$}}
\Procedure{sub}{$u,k$}
  \If{vertex $u$ corresponds to a cluster with a single edge $e$}
	\If{$k=1$}
		\State{\Return{(the label $s$ of the starting point of $e$)}}
	\Else{}
		\State{\Return{(the label $s$ of the ending point of $e$)}}
	\EndIf	
  \Else
  	\State{$v \gets (\mbox{the left child of $u$}$}
  	\State{$w \gets (\mbox{the right child of $u$})$} 
  	\State{$C(v) \gets (\mbox{the size of the cluster of $v$})$}
  	\State{$C(w) \gets (\mbox{the size of the cluster of $w$})$}
  	\If{vertex $u$ is horizontal merge}
  		\If{$1\leq k \leq C(v)$}
  			\State{\Return{\textsc{sub}$(v,k)$}}
  		\Else{}
  			\State{\Return{\textsc{sub}$(w,k-C(v)-1)$}}
		\EndIf
  	\Else{}
  		\State{$D(v)\gets (\mbox{the preorder of the bottom boundary node of the cluster of $v$})$}
  		\If{$1\leq k \leq D(v)$}
  			\State{\Return{\textsc{sub}$(v,k)$}}
  		\ElsIf{$D(v)+1\leq k \leq D(v)+C(w)$}
  			\State{\Return{\textsc{sub}$(w,k-D(v)-1)$}}
  		\Else{}
  			\State{\Return{\textsc{sub}$(v,k-C(w)-1)$}}
  		\EndIf
	\EndIf
  \EndIf
\EndProcedure
\end{algorithmic}
\end{algorithm}

\begin{algorithm}[h]
\caption{$size(x')$: the size of the cluster corresponding to vertex $x'$ of $T'$}
\label{size(x')}
\begin{algorithmic}[1]
\Require{Preorder $x'$}
\Ensure{The size of the cluster for $x'$}
\State{$l  \gets leaf\_rank(leftmost\_leaf(x))$}
\State{$r  \gets leaf\_rank(rightmost\_leaf(x))$}
\State{$l' \gets rank_1(\dummyBV,l-1)+1 $}
\State{$r' \gets rank_1(\dummyBV,r)$}
\State{$k \gets r-l+1$}
\State{$c \gets (r - l+1) - (r' - l') $}
\If{$l' = 0$}
	\State{\Return{$\clustersize[r'] + 2c - k $} }
\Else
	\State{\Return{$\clustersize[r'] - \clustersize[l'-1] + 2c - k $ } }
\EndIf
\end{algorithmic}
\end{algorithm}

\newpage
\section{Components of Top ZDD}

\begin{table}[h]
\caption{Components of the top ZDD}
\label{datastructure}
\begin{tabular}{|c|p{33em}|}
\hline
$bp$ & BP sequence representing the structure of $T'$ \\ \hline
$\dummyBV$ & bitvector showing $i$-th leaf is a dummy vertex or not \\ \hline
$\clustersize$  & array storing cumulative sum of cluster sizes of the first to the $i$-th dummy leaves \\ \hline
$\lspan$ & array storing differences of labels of ending points of $i$-th non-dummy leaf \\ \hline
$\edgetypespan$ & array showing the edge corresponding to the $i$-th non-dummy leaf is $0$-edge or not \\ \hline
$\horizontalBV$ & bitvector showing $i$-th internal vertex is a vertical merge or not \\ \hline
$\pdiff$        & array storing differences of preorders between the top and the bottom boundary nodes
of the vertex corresponding to $i$-th vertical merge \\ \hline
$\ldiff$        & array storing differences of labels between the top and the bottom boundary nodes
of the vertex corresponding to $i$-th vertical merge \\ \hline
$\srcroot$      & bitvector storing in unary codes the number of complement edges from each vertex  \\ \hline
$\dstroot$      & array storing preorders of ending points of the $i$-th complement edge stored in root \\ \hline
$\edgetyperoot$ & array showing the $i$-th complement edge stored in the root is a $0$-edge or not \\ \hline
$\edgeBV$       & bitvector storing in unary codes the number of complement edges from each vertex
stored in the root \\ \hline
$\srcin$        &  array storing local preorders of starting points of $i$-th complement edge stored in non-root \\ \hline
$\dstin$        &  array storing local preorders of ending points of $i$-th complement edge stored in non-root \\ \hline
$\edgetypein$   &  array showing the $i$-th complement edge stored in non-root is $0$-edge or not \\ \hline
\end{tabular}
\end{table}

\newpage
\section{Construction and Traverse Time}

\begin{table}[h]
\centering
\caption {The power set of $\{1,..,A\}$.}
\label{time:powerset}
\begin{tabular}{|c|c|c|c|c|}
\hline
  \multirow{2}{*}{} & \multicolumn{2}{c|}{Construction Time{(}s{)}} & \multicolumn{2}{c|}{Traverse Time {(}$\mu$ s{)}} \\ \cline{2-5} 
 &\ \  top ZDD \ \  & DenseZDD &\ \ top ZDD\ \  & DenseZDD \\ \hline
$A=1000$ & 0.006 & 0.004 & 13.546 & 0.458 \\ \hline
$A=50000$ & 0.217 & 0.116 & 11.768 & 0.198 \\ \hline
\end{tabular}
\end{table}

\begin{table}[h]
\caption {For the set $\{1,..,A\}$ with $A$ elements,
the family of all the set $S$ satisfying 
$(\mbox{The maximum value of $S$})-(\mbox{The minimum value of $S$})\leq B$.}
\label{time:limitedwidth}
\centering
\begin{tabular}{|c|c|c|c|c|}
\hline
  \multirow{2}{*}{} & \multicolumn{2}{c|}{Construction Time{(}s{)}} & \multicolumn{2}{c|}{Traverse Time {(}$\mu$ s{)}} \\ \cline{2-5} 
 &\ \  top ZDD \ \  & DenseZDD &\ \ top ZDD\ \  & DenseZDD \\ \hline
$A=500,B=250$ & 0.264 & 0.078 & 9.082 & 0.244 \\ \hline
$A=1000,B=500$ & 0.776 & 0.412 & 10.419 & 0.229 \\ \hline
\end{tabular}
\end{table}

\begin{table}[h]
\caption{For the set $\{1,..,A\}$ with $A$ elements,
the family of all the sets with cardinality at most $B$.}
\label{time:limitedsize}
\centering
\begin{tabular}{|c|c|c|c|c|}
\hline
  \multirow{2}{*}{} & \multicolumn{2}{c|}{Construction Time{(}s{)}} & \multicolumn{2}{c|}{Traverse Time {(}$\mu$ s{)}} \\ \cline{2-5} 
 &\ \  top ZDD \ \  & DenseZDD &\ \ top ZDD\ \  & DenseZDD \\ \hline
$A=100,B=50$ & 0.011 & 0.006 & 10.892 & 0.320 \\ \hline
$A=400,B=200$ & 0.269 & 0.123 & 16.019 & 0.534 \\ \hline
$A=1000,B=500$ & 2.013 & 0.878 & 20.101 & 0.412 \\ \hline
\end{tabular}
\end{table}

\begin{table}[h]
\caption{Knapsack set families with random weights.
$A$ is the number of elements, $W$ is the maximum weight of an element, $C$ is the capacity of the knapsack.}
\label{time:knapsack}
\centering
\scalebox{0.9}{
\begin{tabular}{|c|c|c|c|c|}
\hline
  \multirow{2}{*}{} & \multicolumn{2}{c|}{Construction Time{(}s{)}} & \multicolumn{2}{c|}{Traverse Time {(}$\mu$ s{)}} \\ \cline{2-5} 
 &\ \  top ZDD \ \  & DenseZDD &\ \ top ZDD\ \  & DenseZDD \\ \hline
{$A=100,W=1000,C=10000$} & 2.974 & 1.210 & 16.716 & 0.259 \\ \hline
{$A=200,W=100,C=5000$} & 2.033 & 1.019 & 23.215 & 0.290 \\ \hline
{$A=1000,W=100,C=1000 $} & 7.010 & 1.481 & 21.698 & 0.534 \\ \hline
{$A=5000,W=100,C=200 $} & 2.084 & 0.954 & 7.365 & 0.519 \\ \hline
{$A=1000,W=10,C=1000$} & 2.597 & 1.712 & 14.127 & 0.244 \\ \hline
{$A=1000,W=100,C=1000$} & 7.010 & 1.481 & 21.698 & 0.534 \\ \hline
\end{tabular}
}
\end{table}

\begin{table}[h]
\caption{The family of edge sets which are matching of a given graph.}
\label{time:matching}
\centering
\begin{tabular}{|c|c|c|c|c|}
\hline
  \multirow{2}{*}{} & \multicolumn{2}{c|}{Construction Time{(}s{)}} & \multicolumn{2}{c|}{Traverse Time {(}$\mu$ s{)}} \\ \cline{2-5} 
 &\ \  top ZDD \ \  & DenseZDD &\ \ top ZDD\ \  & DenseZDD \\ \hline
 $8\times 8$ grid & 0.030 & 0.020 & 11.678 & 1.053 \\ \hline
complete graph $K_{12}$ & 0.019 & 0.009 & 14.864 & 0.290 \\ \hline
“{ \itshape{Interoute}}” & 0.028 & 0.016 & 15.588 & 0.397 \\ \hline
\end{tabular}
\end{table}

\begin{table}[h]
\caption{Set families of frequent item sets.}
\label{time:frequent}
\centering
\begin{tabular}{|c|c|c|c|c|}
\hline
  \multirow{2}{*}{} & \multicolumn{2}{c|}{Construction Time{(}s{)}} & \multicolumn{2}{c|}{Traverse Time {(}$\mu$ s{)}} \\ \cline{2-5} 
 &\ \  top ZDD \ \  & DenseZDD &\ \ top ZDD\ \  & DenseZDD \\ \hline
“{\small\itshape{mushroom}}” $(p=0.001)$ & 0.093 & 0.037 & 14.100 & 0.198 \\ \hline
“{\small\itshape{retail}}” $(p=0.00025)$ & 0.099 & 0.134 & 12.857 & 0.702 \\ \hline
“{\small\itshape{T40I10D100K}}” $(p=0.005)$ & 0.198 & 0.117 & 13.788 & 0.183 \\ \hline
\end{tabular}
\end{table}

\begin{table}[h]
\caption{Families of paths in $n\times n$ grid graph.}
\label{time:gridpath}
\centering
\begin{tabular}{|c|c|c|c|c|}
\hline
  \multirow{2}{*}{} & \multicolumn{2}{c|}{Construction Time{(}s{)}} & \multicolumn{2}{c|}{Traverse Time {(}$\mu$ s{)}} \\ \cline{2-5} 
 &\ \  top ZDD \ \  & DenseZDD &\ \ top ZDD\ \  & DenseZDD \\ \hline
$n=6$ & 0.022 & 0.011 & 15.491 & 0.793 \\ \hline
$n=7$ & 0.082 & 0.036 & 12.039 & 1.022 \\ \hline
$n=8$ & 0.536 & 0.153 & 12.229 & 1.144 \\ \hline
$n=9$ & 1.821 & 0.944 & 14.233 & 1.404 \\ \hline
\end{tabular}
\end{table}

\begin{table}[h]
\caption{Families of solutions of the $n$-queen problem.}
\label{time:nqueen}
\centering
\begin{tabular}{|c|c|c|c|c|}
\hline
  \multirow{2}{*}{} & \multicolumn{2}{c|}{Construction Time{(}s{)}} & \multicolumn{2}{c|}{Traverse Time {(}$\mu$ s{)}} \\ \cline{2-5} 
 &\ \  top ZDD \ \  & DenseZDD &\ \ top ZDD\ \  & DenseZDD \\ \hline
$n=11$ & 0.038 & 0.015 & 17.184 & 0.778 \\ \hline
$n=12$ & 0.335 & 0.065 & 21.581 & 0.900 \\ \hline
$n=13$ & 1.722 & 0.419 & 20.173 & 1.099 \\ \hline
\end{tabular}
\end{table}

\end{document}